\documentclass[smallabstract,smallcaptions]{dccpaper}
\usepackage[numbers,sort&compress]{natbib}
\renewcommand*{\citet}[2][]{{\cite[#1]{#2}}}

\usepackage{comment}
\usepackage[ruled,vlined,lined,boxed,commentsnumbered]{algorithm2e}

\usepackage{amsthm} 
\usepackage{thmtools}
\usepackage{tikz}
\usetikzlibrary{automata, positioning, arrows}

\newtheorem{theorem}{Theorem}
\newtheorem{lemma}[theorem]{Lemma}
\newtheorem{definition}[theorem]{Definition}

\usepackage{citesort}
\usepackage{hyperref}

\newlength{\figurewidth}
\newlength{\smallfigurewidth}
\usepackage{graphicx}
\usepackage{xcolor}
\usepackage{booktabs}
\usepackage[capitalise]{cleveref}
\usepackage{caption}
\usepackage{subcaption}

\title{\large \textbf{Linear-time Minimization of Wheeler DFAs}}
\author{%
Jarno Alanko$^{1,2}$, Nicola Cotumaccio$^{3}$, 
Nicola Prezza$^{4}$ \\[0.5em]
{\small\begin{minipage}{\linewidth}\begin{center}
$^{1}$Dept. of Computer Science, University of Helsinki, Finland, \url{jarno.alanko@helsinki.fi} \\
$^{2}$Faculty of Computer Science, Dalhousie University, Halifax, Canada \\
$^{3}$GSSI, L'Aquila, Italy, \url{nicola.cotumaccio@gssi.it}   \\
$^4$ DAIS, Ca' Foscari University, Venice, Italy, \url{nicola.prezza@unive.it} \\
\end{center}\end{minipage}}
}

\begin{document}

\maketitle              
\vspace{2em}

\thispagestyle{empty}

\begin{abstract}
Wheeler DFAs (WDFAs) are a sub-class of finite-state automata which is playing an important role in the emerging field of \emph{compressed data structures}: as opposed to general automata, WDFAs can be stored in just $\log\sigma + O(1)$ bits per edge, $\sigma$ being the alphabet's size, and support optimal-time pattern matching queries on the substring closure of the language they recognize. An important step to achieve further compression is minimization. When the input $\mathcal A$ is a general deterministic finite-state automaton (DFA), the state-of-the-art is represented by the classic Hopcroft's algorithm, which runs in $O(|\mathcal A|\log |\mathcal A|)$ time. This algorithm stands at the core of the only existing minimization algorithm for Wheeler DFAs, which inherits its complexity. In this work, we show that the minimum WDFA equivalent to a given input WDFA can be computed in linear $O(|\mathcal A|)$ time. 
When run on de Bruijn WDFAs built from real DNA datasets, an implementation of our algorithm reduces the number of nodes from 14\% to 51\% at a speed of more than 1 million nodes per second. 
\end{abstract}
\Section{Introduction}

Minimizing deterministic finite-state automata (DFA) is a classic problem in computer science. The most well-known method solving this problem is the partition refinement algorithm of Hopcroft \cite{hopcroft1971n}, which runs in time $O(|\mathcal A| \log | \mathcal A |)$, where $ |\mathcal A|$ denotes the size of the input automaton. Revuz improved this to linear time for \emph{acyclic} DFAs \cite{revuz1992minimisation}, but no linear-time algoritihm for the general case is known to date. 


In this work, we consider the problem of minimizing Wheeler DFAs (WDFA): given a WDFA, compute the minimum equivalent WDFA. This is a class of automata recently introduced by Gagie et al. \cite{GAGIE201767}
which is gaining a lot of interest in the field of compressed data structures. 
A Wheeler automaton with an alphabet of size $\sigma$ can be represented in just $\log \sigma + O(1)$ bits per edge, and indexed for substring queries with the addition of auxiliary succinct data structures.
Practically-relevant classes of labeled graphs that are always Wheeler include disjoint paths (sets of strings) \cite{mantaci2007extension}, trees \cite{XBWT}, and de Bruijn graphs \cite{BOSS}.

Wheeler DFAs are popular in indexing genomic databases. For example, the compressed index structure of the variation graph toolkit VG \cite{vg} is based on the WDFA of a compacted de Bruijn graph of the input data. Other applications include the colored de Bruijn graph indexes VARI \cite{VARI} and Themisto \cite{themisto}. Efficient WDFA minimization has applications in compressing the space of all these data structures. Note that regular DFA minimization algorithms are not suitable for the problem, because the minimum DFA equivalent to a given WDFA is not necessarily Wheeler (and thus it is not compressible/indexable as efficiently).

\subsection*{Related Work}

There are elements of WDFA minimization in the variation graph indexing tools GCSA and GCSA2 \cite{siren2014indexing, siren2017indexing} used in the VG toolkit \cite{vg}. 
Due to historical reasons, in both GCSA and GCSA2 the transitions of the automata travel in the reverse direction compared to the convention in Wheeler graphs. For clarity, in this section we describe GSCA and GSCA2 using the orientation convention in Wheeler graphs.

GCSA turns a determinized acyclic variation graph $G$ into an equivalent Wheeler DFA $W$ by splitting states until the graph satisfies the properties of a Wheeler graph. During the process, states of $W$ with the same incoming path label sets are merged if they correspond to the same original node in $G$. By construction, this merging preserves the language of the automaton, but it does not merge equivalent states that correspond to distinct nodes in the variation graph.

The GCSA2 data structure builds the WDFA 
of the order-$k$ de Bruijn graph  of the variation graph. 
A de Bruijn graph  (dBg) of a string set is a graph where the nodes represents all substrings of length $k$ ("$k$-mers") of the input. There is an edge from node $u$ to node $v$ if there is a $(k+1)$-mer in the input that is suffixed by the $k$-mer of $u$ and prefixed by the $k$-mer of $v$. The label of the edge is the last character of $v$. 
In GCSA2, the dBg is indexed and minimized similarly to GCSA. As a final step, a more aggressive form of merging is applied --- targeting nodes sharing some incoming suffix --- such that the nodes map to the same \emph{set} of nodes in the variation graph. This may be seen as reducing the order $k$ of the dBg locally.

A core issue with the minimization in GCSA and GCSA2 is that the merging is not based on the fundamental Myhill-Nerode equivalence relation of states \cite{nerode1958linear}. Instead, they on rely on the mapping to the variation graph to detect equivalent states. This acts as a proxy to Myhill-Nerode equivalence, but the method is more restrictive and may fail to detect all equivalent states.

In 2020, Alanko et al. \cite{alankosoda2020} formalized the WDFA minimization problem using the Myhill-Nerode equivalence relation. They characterized the minimum WDFA of an automaton and gave an algorithm to minimize a WDFA in time $O(|\mathcal A| \log |\mathcal A|)$, using the classic Hopcroft's DFA minimization algorithm \cite{hopcroft1971n} as a subroutine. This algorithm stands as the fastest currently known algorithm for WDFA minimization.

In the special case of de Bruijn WDFAs, there is another minimization technique available. 
To turn the de Bruijn graph into a WDFA, additional technical dummy nodes need to be added: these represent prefixes of $k$-mers 
being themselves prefixes of strings in the input set. 
Alanko et al. \cite{alanko2021buffering} showed how to eliminate redundant dummy nodes from the automaton. This may change the language of the automaton but it does not change the substrings of length $k$ of the language.

\subsection*{Our contributions}

In this work, we improve on the WDFA minimization algorithm of Alanko et al. \cite{alankosoda2020}, by eliminating the dependency on Hopcroft's algorithm, and instead exploiting the special structure of a Wheeler automaton to minimize it in \emph{linear time}.

Experimentally, we use our WDFA minimization algorithm to show that the order-$k$ de Bruijn WDFA contains significant redundancies. In particular, we show that on real DNA sequence data, we can compress it by up to 51\%, while maintaining the language of the automaton and the Wheeler property.

\Section{Notation and Preliminaries}

We denote DFAs as quintuples 
$\mathcal{A} = (Q, \Sigma, \delta, s, F)$, where $Q$ is the set of states, $\Sigma = \{0,\dots, |Q|^{O(1)}\}$ is an alphabet of polynomial size, 
$\delta : Q\times \Sigma \rightarrow Q$ is the transition function, $s\in Q$ is the initial (start) state, and $F \subseteq Q$ is the set of final states. 

Given a state of $\mathcal A$, we will denote with $final(u)$ the boolean predicate returning \texttt{true} if and only if $u\in F$.

Note that a DFA $\mathcal{A} = (Q, \Sigma, \delta, s, F)$ can be interpreted as an edge-labeled graph $(Q,E)$, where $E \subseteq Q\times Q \times \Sigma$ is the set of labeled edges  $E = \{(q,q',a)\ |\ \delta(q,a)=q'\}$. We indicate with $|\mathcal A| = |Q| + |E|$ the size of this graph. We denote with $out(u) = \{a\in \Sigma\ |\ \delta(u,a) \in Q\}$ the set of all labels exiting $u$.

We do not assume $\delta$ to be complete: there could exist pairs $q\in Q$, $a\in \Sigma$ for which $\delta(q,a)$ is not defined. From a graph-theoretic perspective this means that, in the labeled graph $(Q,E)$ underlying $\mathcal A$, a state $q$ does not necessarily have one outgoing edge for each alphabet's character. We also assume that all states in $Q \setminus \{s\}$ are reachable from $s$, that there are no incoming edges in $s$, and  that every state is either final or it allows to reach a final state.

The \emph{regular language recognized by a DFA $\mathcal{A}$}, denoted $\mathcal L(\mathcal{A})$, is the set of all strings labeling paths in $(Q,E)$ which connect $s$ with some state in $F$. More formally, let
$\hat\delta : Q\times \Sigma^* \rightarrow Q$ be the extension of $\delta$ defined as $\hat \delta (q, \epsilon ) = q $, where $ \epsilon $ is the empty string, and $\hat\delta(q,\alpha\cdot a) = \delta(\hat\delta(q,\alpha),a)$ for $\alpha \in \Sigma^*$.

Then, $\mathcal L(\mathcal{A}) = \{\alpha \in \Sigma^* \ |\ (\exists q\in F)(\hat\delta(s,\alpha)=q)\}$.

We say that an equivalence relation $ \sim $ on $ Q $ is \emph{right-invariant} if and only if for every $ u, v \in Q $ such that $ u \sim v $ and for every $ a \in \Sigma $, $ \delta (u, a) $ is defined if and only if $ \delta (v, a) $ is defined, and, if so, it holds $ \delta (u, a) \sim \delta (v, a) $. We denote the equivalence class of $v$ with $[v]_\sim$.

Let $ \sim $ be a right invariant equivalence relation on a DFA $\mathcal A $. The \emph{quotient automaton} is defined as $ \mathcal A_{/\sim} = (Q_\sim, \Sigma, \delta_\sim, [s]_\sim, F_\sim)  $, where $ Q_\sim = \{[u]_\sim\ |\  u \in Q \} $, $ \delta_\sim([u]_\sim, a) = [v]_\sim $ if and only if $ \delta (u, a) = v $, and $ F_\sim = \{[u]_\sim\ |\ u \in F \} $. A classic result is that, since $ \sim $ is right-invariant, then $ A_{/\sim} $ is a well-defined DFA such that $ \mathcal{L(A_{/\sim})} = \mathcal{L(A)} $.

We define a special equivalence relation $ \approx_\mathcal{A} $ on $ Q $ as follows: $ u \approx_\mathcal{A} v$ if and only if for every $ \alpha \in \Sigma^* $ we have that $\hat\delta(u,\alpha) \in F$ if and only if $\hat\delta(v,\alpha) \in F$. 
In other words, $ \approx_\mathcal{A} $ is the state version of the classic \emph{Myhill-Nerode} equivalence relation \cite{nerode1958linear}. Note that $ \approx_\mathcal{A} $ is right-invariant. The fundamental \emph{Myhill-Nerode Theorem} \cite{nerode1958linear} states that  $\mathcal A_{/\approx_\mathcal{A}}$ is the minimum (in the number of states) DFA recognizing $\mathcal L(\mathcal A)$. 
The minimum DFA $\mathcal A_{/\approx_\mathcal{A}}$ can be computed in $O(|\mathcal A|\log |\mathcal A|)$ time by means of a classic algorithm described by Hopcroft \cite{hopcroft1971n}.

We assume that elements of $\Sigma$ are totally-sorted according to the standard integer order, which we denote here by $\preceq$ when referring to elements of $\Sigma$.

A \emph{Wheeler DFA} (WDFA for brevity) \cite{GAGIE201767,alankosoda2020} $\mathcal A$ is a DFA for which there exists a \emph{total} order $\le\ 
\subseteq Q\times Q$ (called \emph{Wheeler order}) satisfying the following three axioms (in the following, $ u < v $ means $ u \le v $ and $ u \not = v $):

\begin{itemize}
    \item[] \emph{(i)} $s \leq u$ for every $u\in Q$.
    \item[] \emph{(ii)} If $u' = \delta(u,a)$, $v' = \delta(v,b)$, and $a \prec b$, then $u'<v'$.
    \item[] \emph{(iii)} If $u' = \delta(u,a)$, $v' = \delta(v,b)$, $a = b$, and $u<v$, then $u'\le v'$.
\end{itemize}


In \cite{alankosoda2020} it was showed that (1) if $ \mathcal{A} $ is a WDFA, then a Wheeler order $ \le $ on $ \mathcal{A} $ is uniquely determined (that is, $ \le $ is \emph{the} Wheeler order on $ \mathcal{A} $), and (2) if $ \mathcal{L} $ is a regular language recognized by some WDFA, then a WDFA $ \mathcal{A'} $ recognizing $ \mathcal{L} $ and having the minimum number of states is unique up to isomorphism (that is, $ \mathcal{A}' $ is the \emph{minimum} WDFA recognizing $ \mathcal{L} $).

Wheeler axioms imply the \emph{input-consistency} property: if $\delta(q,a) = \delta(q',a')$, then $a=a'$. From a graph-theoretic perspective this means that, in the labeled graph $(Q,E)$ underlying $\mathcal A$, all edges entering the same state bear the same label.
With $\lambda(u) = a$, for $u\in Q$, we denote the unique $a\in \Sigma$ with $\delta(u',a)=u$ for all predecessors $u'$ of $u$. For $s$ we use the convention $\lambda(s) = \# \notin \Sigma$, where $\# \prec a$ for all $a\in \Sigma$.

\Section{Linear-time Minimization of WDFAs}



The Wheeler DFA minimization problem was first addressed by Alanko et al. in \cite{alankosoda2020}:

\vspace{5pt}
\noindent
\textbf{Problem: WDFA minimization} \cite{alankosoda2020}. \emph{Given a WDFA $\mathcal A$, compute the smallest (minimum number of states) WDFA $\mathcal A'$ such that $\mathcal L(\mathcal A') = \mathcal L(\mathcal A)$.}

\vspace{5pt}
We note that it is not important whether or not the input WDFA is sorted (that is, whether or not the Wheeler order $\le$ is given as a part of the input), since WDFAs can be sorted in linear time \cite{alankosoda2020} (assuming a polynomial integer alphabet, as we do in the present paper).

Alanko et al. in \cite{alankosoda2020} presented an algorithm solving the WDFA minimization problem in $O(|\mathcal A|\log |\mathcal A|)$ time. The main bottleneck of this algorithm is represented by a call to Hopcroft's algorithm. After the Myhill-Nerode equivalence classes of $\mathcal A$ ($[u]_{\approx_{\mathcal A}}$, for $u\in Q$) have been computed, the minimum WDFA $\mathcal A'$ can be derived in linear time by means of the following lemma:

\begin{lemma}[minimum WDFA \cite{alankosoda2020}]\label{lem:minimum WDFA}
	Let $ \mathcal{A} = (Q, \Sigma, \delta, s, F)  $ be a Wheeler DFA, let $\le$ be the Wheeler order on $ \mathcal{A} $, and write $ Q = \{u_1, u_2, \dots, u_n \} $, with $ u_1 < u_2 < \dots < u_n $. Let $ \equiv_\mathcal{A} $ be the equivalence relation on $ Q $ that puts in the same equivalence classes exactly all states belonging to the maximum runs of states $ u_i, u_{i + 1}, \dots, u_{i + t} $ such that: (1) $ \lambda (u_i) = \lambda (u_{i + 1}) = \dots = \lambda (u_{i + t}) $, and (2) $ u_i \approx_\mathcal{A} u_{i + 1} \approx_\mathcal{A} \dots \approx_\mathcal{A} u_{i + t} $.
	Then, $ \equiv_\mathcal{A} $ is right-invariant and $ \mathcal{A}_{/\equiv_\mathcal{A}} $ is the minimum WDFA recognizing $ \mathcal{L(A)} $.
\end{lemma}

In other words, in order to find the smallest WDFA one needs to identify maximal runs of consecutive states (in Wheeler order) that are Myhill-Nerode equivalent and that are reached by the same label. 
The bottleneck of a direct implementation of this procedure 
is the computation of $\approx_\mathcal{A}$, which in general takes  $O(|\mathcal A|\log |\mathcal A|)$ time using Hopcroft's algorithm.

\subsection*{Our Algorithm}

In this section we present a linear-time algorithm for the WDFA minimization problem. 
The idea behind our algorithm is to use the characterization of the minimum WDFA provided by Lemma \ref{lem:minimum WDFA} and exploit the following observation: since classes of $\equiv_\mathcal{A}$ form intervals in the Wheeler order, it is sufficient to identify \emph{borders} $(u_i,u_{i+1})$ between these intervals in order to reconstruct $\equiv_\mathcal{A}$ (and thus the minimum WDFA). 
Let $u_1 < u_2 < \dots < u_n$ be the Wheeler order. 
The borders between classes of $\equiv_\mathcal{A}$ can be found efficiently by exploiting the properties stated in the following two lemmas:

\begin{lemma}\label{lem:adjacent}
For any string $\alpha$, if $v = \hat\delta(u_i,\alpha)$,  $v' = \hat\delta(u_{i+1},\alpha)$ are both defined and $v\neq v'$, then $v = u_j$ and $v' = u_{j+1}$ for some $1 \leq j < n$, that is, $v$ and $v'$ are also adjacent in the Wheeler order.
\end{lemma}
\begin{proof}
Without loss of generality, we can assume that $\alpha = a \in \Sigma$ is a character (the claim will follow by extension). 
By Wheeler Axiom (iii) and since $v\neq v'$, it must be $v < v'$.
Assume, for a contradiction, that $v=u_j$ and $v'=u_{j'}$ with $j' > j+1$. Let therefore $v''$ be any node such that $v < v'' < v'$. By Wheeler Axiom (ii), it must be $\lambda(v) = \lambda(v'') = \lambda(v') = a$. Let therefore $u'$ be any $a$-predecessor of $v''$: $\delta(u',a) = v''$. Since $\mathcal A$ is a DFA and $v \neq v''\neq v'$, it must be the case that $u_i \neq u' \neq u_{i+1}$. Since we also know that $u_i$ and $u_{i+1}$ are adjacent in Wheeler order, we therefore have two cases: either $u'<u_i$ or $u'>u_{i+1}$. If $u'<u_i$, then by Wheeler Axiom (iii) it must be $v'' < v$, a contradiction. Similarly, if $u'>u_{i+1}$ then by Wheeler Axiom (iii) it must be $v'' > v'$, again a contradiction.
\end{proof}

The second property follows directly from the right-invariance of $\equiv_\mathcal{A}$: 

\begin{lemma}\label{lem:borders}
If $u_j \not\equiv_\mathcal{A} u_{j+1}$ (i.e. $(u_j,u_{j+1})$ is a border), then all pairs $(u_i,u_{i+1})$ such that $\hat\delta(u_i,\alpha) = u_j$ and $\hat\delta(u_{i+1},\alpha) = u_{j+1}$ for some string $\alpha\in\Sigma^*$, are also such that $u_{i} \not\equiv_\mathcal{A} u_{i+1}$ (that is, they are borders). 
\end{lemma}

These properties naturally suggest a linear-time reachability algorithm for the problem: first, mark all ``base-case'' borders (below we formalize this notion). Then, mark also all borders $(u_i,u_{i+1})$ that lead to a marked ``base-case'' border $(u_j,u_{j+1})$ through some string $\alpha$, that is, such that $\hat\delta(u_i,\alpha) = u_j$ and $\hat\delta(u_{i+1},\alpha) = u_{j+1}$. 
Lemma \ref{lem:adjacent} guarantees that for every string $ \alpha $ states $\hat\delta(u_i,\alpha)$ and $\hat\delta(u_{i+1},\alpha)$, if distinct, are indeed adjacent in Wheeler order.

In order to formalize this reasoning, let us define the \emph{border graph} of a WDFA: 

\begin{definition}[border graph of a WDFA]
Let $\mathcal A$ be a WDFA. The border graph of $\mathcal A$ is the (unlabeled) graph $ \mathcal B(\mathcal A) = (B, Z)$ where  $ B = \{(u_i, u_{i + 1})\ |\ 1 \le i < n, \lambda (u_i) = \lambda (u_{i + 1}) \} $ and $ Z = \{((u_i, u_{i + 1}), (u_j, u_{j + 1})) \in B \times B\ |\  u_i = \delta (u_j, \lambda (u_i))\ \wedge\ u_{i + 1} = \delta (u_{j + 1}, \lambda (u_i)) \} $.
\end{definition}

In other words, an edge of $ \mathcal B(\mathcal A)$ exists between borders $(u_i, u_{i + 1})$ and $(u_j, u_{j + 1})$ whenever $u_i$ (respectively, $u_{i+1}$) can be reached by $u_j$ (respectively, $u_{j+1}$) by an edge labeled $a = \lambda(u_i) = \lambda(u_{i+1})$.

In general, $ \mathcal B(\mathcal A)$ may contain cycles. 
With the next lemma we put a bound to the size of $ \mathcal B(\mathcal A) $, by showing that the maximum out-degree in the graph is at most 1.

\begin{lemma}\label{lem: border graph}
$ \mathcal B(\mathcal A) $ has at most $n-1$ edges and $n-1$ vertices, where $n$ is the number of states of $\mathcal A$. Moreover, $ \mathcal B(\mathcal A) $  can be constructed in $O(|\mathcal A|)$ time given $\mathcal A$ as input. 
\end{lemma}
\begin{proof}
Clearly, $|B| \leq n-1$ since elements of $B$ are pairs of adjacent (in Wheeler order) states of $\mathcal A$. 

We now show that for every $ (u_i, u_{i + 1}) \in B $, there exists at most one $ (u_j, u_{j + 1}) \in B $ such that $ ((u_i, u_{i + 1}), (u_j, u_{j + 1})) \in Z $.
Indeed, if $ u_r, u_s \in Q $ are any states such that $ u_i = \delta (u_r, a) $ and $ u_{i + 1} = \delta (u_s, a) $, where $ a = \lambda (u_i) = \lambda (u_{i + 1}) $, then from $ u_i < u_{i + 1} $ and from Wheeler Axiom (iii) it follows $ u_r < u_s $ (equality cannot hold, $\mathcal A$ being a DFA), and $ (u_r, u_s) \in B$ if and only if $ r $ is the largest integer such that $ u_i = \delta (u_r, a) $, $ s $ is the smallest integer such that $ u_{i + 1} = \delta (u_s, a) $, $ s = r + 1 $ and $ \lambda (u_r) = \lambda (u_s) $. In other words, $ |Z| \le |B| \le n - 1 $.

Finally,  $ \mathcal B(\mathcal A) $ can be built in $ O (|\mathcal A|) $ time as follows. 
Consider the list $u_1 < \dots < u_n$ of $\mathcal A$'s states, sorted in Wheeler order (sorting WDFAs takes linear time \cite{Alanko2020Wheeler}).
For each $(u_i,u_{i+1})$ with $\lambda(u_i) = \lambda(u_{i+1})$ and each letter $a$ labeling a transition leaving $u_i$, let $v = \delta(u_i,a)$ and $v' = \delta(u_{i+1},a)$ (note that the outgoing edges of each node can be sorted in linear time by their label to speed up this operation). If both $v$ and $v'$ exist and are distinct, they must indeed be adjacent in Wheeler order by Lemma \ref{lem:adjacent}: $v=u_j$ and $v' = u_{j+1}$, for some $1 \leq j < n$. Then, insert in $\mathcal B(\mathcal A)$ an edge $((u_j,u_{j+1}),(u_i,u_{i+1}))$.
\end{proof}

We describe our minimization algorithm as Algorithm \ref{alg:minimize}. In Line \ref{algline:sort}, we compute in linear time the Wheeler order on $\mathcal A$ by using the algorithm of Alanko et al. \cite{alankosoda2020}. In line \ref{algline:border} we compute the border graph  $(B,Z) = \mathcal B(\mathcal A)$ of $\mathcal A$. This is done in linear time, see Lemma \ref{lem: border graph}. In Lines 3-5, we mark base-case nodes: for every pair of adjacent nodes, if they are not both final/not final, or if their sets of outgoing labels are not equal then they cannot be Myhill-Nerode equivalent (and thus $\equiv_\mathcal{A}$-equivalent). This step takes time proportional to the number of edges of $\mathcal A$. In Line 
\ref{algline:mark reachable} we perform a linear-time visit of $\mathcal B(\mathcal A)$ starting from the nodes that have been marked in Lines 3-5. During this visit, we mark every visited node. 
In Lines 7-9 we compute the equivalence classes of $\equiv_\mathcal{A}$. In Line 8, the predicate $marked((u_i,u_{i+1}))$ returns \texttt{true} if and only if $(u_i,u_{i+1}) \in B$ has been marked in the previous lines. Procedure make\_equivalent($u_i, u_{i+1}$) at Line \ref{algline:equiv} records that nodes $u_i, u_{i+1}$ belong to the same equivalence class of $\equiv_\mathcal{A}$. 
To conclude, at Line \ref{algline:return} we return the quotient automaton $\mathcal A_{/\equiv_\mathcal{A}}$ which, by Theorem \ref{thm:linear minimize}, is the minimum WDFA recognizing $\mathcal L(\mathcal A)$. The  WDFA $\mathcal A_{/\equiv_\mathcal{A}}$ can be computed in linear time by collapsing each equivalence class of $\equiv_\mathcal{A}$ (intervals in Wheeler order) into one state and deduplicating equally-labeled edges exiting the same equivalence class.

\begin{algorithm}[h!]
    \caption{minimize($\mathcal A$)}\label{alg:minimize}
	\SetKwInOut{Input}{input}
	\SetKwInOut{Output}{output}
	\SetSideCommentLeft
	\LinesNumbered
	
	\Input{A WDFA $\mathcal A$}
	\Output{The minimum WDFA $\mathcal A'$ such that $\mathcal L(\mathcal A) = \mathcal L(\mathcal A')$}
    
    \BlankLine
    
    $<\  \leftarrow$ sort($\mathcal A$)\tcp*{Compute Wheeler order $u_1 < \dots < u_n$ of $\mathcal A$} \label{algline:sort}
    $(B,Z) \leftarrow \mathrm{border\_graph}(\mathcal A, <)$\tcp*{Compute border graph $\mathcal B(\mathcal A)$ of  $\mathcal A$} \label{algline:border}

    \BlankLine

    \For{$(u_i,u_{i+1}) \in B$}{
    
        \If{$out(u_i) \neq out(u_{i+1}) \vee final(u_i) \neq final(u_{i+1})$}{
        
            mark($(u_i,u_{i+1})$)\tcp*{Mark a "base-case" node of $\mathcal B(\mathcal A)$} \label{algline:mark base}
        
        }
    
    }

    \BlankLine
    
    mark\_reachable($B,Z$)\tcp*{Propagate "base-case" marked nodes} \label{algline:mark reachable}
    
    \BlankLine
    
    \For{$i = 1, \dots, n-1$}{
    
        \If{$\Big(\mathtt{not}\ \mathrm{marked}\big((u_i,u_{i+1})\big)\Big) \wedge \lambda(u_i) = \lambda(u_{i+1})$}{
        
            make\_equivalent($u_i, u_{i+1}$)\tcp*{Record that $u_i \equiv_{\mathcal A} u_{i+1}$} \label{algline:equiv}
        
        }
    
    }
    
    \BlankLine

    \Return $\mathcal A_{/\equiv_{\mathcal A}}$\tcp*{Compute and return quotient automaton} \label{algline:return}

\end{algorithm} 

In Figure \ref{fig:example} we pictorially show how Algorithm \ref{alg:minimize} minimizes a WDFA.

\begin{theorem}\label{thm:linear minimize}
	Let $\mathcal{A}$ be a WDFA. Algorithm \ref{alg:minimize} computes the minimum WDFA recognizing $ \mathcal{L(A)} $ in $O(|\mathcal A|)$ time.
\end{theorem}

\begin{proof}
Complexity follows from the algorithm's description: all steps take linear time.

By Lemma \ref{lem:minimum WDFA}, our claim will follow if we prove that $ u_i \not \equiv_\mathcal{A} u_{i + 1} $ if and only if either $ \lambda (u_i) \not = \lambda (u_{i + 1}) $ or $ (u_i, u_{i + 1}) $ is marked in $\mathcal B(\mathcal A)$.

($\Leftarrow$) Suppose that $\lambda(u_i) \neq \lambda(u_{i+1})$. Then, $u_i \not \equiv_\mathcal{A} u_{i + 1}$ follows immediately by definition of $\equiv_\mathcal{A}$. The other case to consider is when
$(u_i, u_{i + 1})$ is marked. 
Then, by the definition of $\mathcal B(\mathcal A)$ this means that there exists a pair $(u_j, u_{j + 1})$ (possibly, $i=j$) that was marked in Line \ref{algline:mark base} (in particular, $u_j \not \approx_\mathcal{A} u_{j + 1}$ and so $u_j \not \equiv_\mathcal{A} u_{j + 1}$) such that $(u_i, u_{i + 1})$ is reachable from $(u_j, u_{j + 1})$ in $\mathcal B(\mathcal A)$. In turn, by the definition of $\mathcal B(\mathcal A)$ this implies that there exists a string $\alpha$ such that $\hat\delta(u_i,\alpha) = u_j$ and $\hat\delta(u_{i+1},\alpha) = u_{j+1}$. By  Lemma \ref{lem:borders}, 
this implies that $u_i \not \equiv_\mathcal{A} u_{i + 1}$.

($\Rightarrow$) Conversely, suppose $ u_i \not \equiv_\mathcal{A} u_{i + 1} $. Then, either $\lambda(u_i) \neq \lambda(u_{i+1})$ (and the claim follows), or $\lambda(u_i) = \lambda(u_{i+1})$. In the latter case, by definition of $\equiv_\mathcal{A}$ it must be $ u_i \not \approx_\mathcal{A} u_{i + 1}$. 
Then, there must exist two states $v,v'$ such that $\hat\delta(u_i,\alpha) = v$ and $\hat\delta(u_{i+1},\alpha) = v'$ for some string $\alpha$, with either $out(v) \neq out(v')$ or $final(v) \neq final(v')$ (in particular, $v\neq v'$). Indeed, let $ \alpha' $ be a shortest string witnessing that $ u_i \not \approx_\mathcal{A} u_{i + 1}$. If $ \hat\delta(u_i,\alpha') $ and $ \hat\delta(u_{i+1},\alpha') $ are both defined, let $ v = \hat\delta(u_i,\alpha') $, $  v' = \hat\delta(u_{i+1},\alpha') $ and $ \alpha = \alpha' $ (in this case $final(v) \neq final(v')$). If exactly one between $ \hat\delta(u_i,\alpha') $ and $ \hat\delta(u_{i+1},\alpha') $ is not defined, then $ \alpha' $ is not the empty string, so we can write $ \alpha' = \alpha'' a $ with $ \alpha'' \in \Sigma^* $ and $ a \in \Sigma $, where both $ \hat\delta(u_i,\alpha'') $ and $ \hat\delta(u_{i+1},\alpha'') $ are defined (by the minimality of $ \alpha' $), so let $ v = \hat\delta(u_i,\alpha'') $, $  v' = \hat\delta(u_{i+1},\alpha'') $ and $ \alpha = \alpha'' $ (in this case $out(v) \neq out(v')$). 
From Lemma \ref{lem:adjacent}, $v$ and $v'$ must be adjacent in Wheeler order, i.e. $v=u_j$ and $v'=u_{j+1}$ for some $1 \leq j < n$.
This implies that (i) $(u_j,u_{j+1})$ is marked in Line \ref{algline:mark base} and (ii)  $(u_i,u_{i+1})$ is reachable from  $(u_j,u_{j+1})$ in $\mathcal B(\mathcal A)$. Finally, (i) and (ii) imply that $(u_i,u_{i+1})$ is marked during the visit of $\mathcal B(\mathcal A)$ in Line \ref{algline:mark reachable}.
\end{proof}

\begin{figure}
   \begin{subfigure}{0.24\textwidth}
   \centering
   \scalebox{0.4}{
	\begin{tikzpicture}[shorten >=1pt,node distance=2cm,on grid,auto]
	\node[state, initial, fill = purple] (1) {$1$};
	\node[state, fill= orange] (2) [above of=1] {$2$};
	\node[state, accepting, fill = green, very thick] (10) [right of = 2] {$10$};
	\node[state, fill = green] (9) [right of = 1] {$9$};
	\node[state, fill = orange] (4) [right of = 10] {$4$};
	\node[state, accepting, fill = green, very thick] (11) [right of = 4] {$11$};
	\node[state, fill = orange] (5) [right of = 11] {$5$};
	\node[state, accepting, fill = cyan, very thick] (22) [right of = 5] {$22$};
	\node[state, fill = pink] (19) [below of = 1] {$19$};
	\node[state, fill = orange ] (7) [right of = 19] {$7$};
	\node[state, fill = orange] (3) [right of = 9] {$3$};
	\node[state, fill = green] (16) [below of = 7] {$16$};
	\node[state, fill = green] (13) [right of = 7] {$13$};
	\node[state, accepting, fill = green, very thick] (12) [right of = 3] {$12$};
	\node[state, accepting, fill = orange, very thick] (6) [right of = 12] {$6$};
	\node[state, accepting, fill = cyan, very thick] (23) [right of = 6] {$23$};
	\node[state, fill = orange] (8) [below of = 16] {$8$};
	\node[state, fill = cyan] (26) [left of = 8] {$26$};
	\node[state, fill = pink] (20) [below of = 13] {$20$};
	\node[state, fill = green] (17) [right of = 20] {$17$};
	\node[state, fill = green] (14) [below of = 6] {$14$};
	\node[state, accepting, fill = cyan, very thick] (24) [below of = 12] {$24$};
	\node[state, fill = cyan] (27) [below of = 20] {$27$};
	\node[state, accepting, fill = pink, very thick] (21) [right of = 27] {$21$};
	\node[state, fill = green] (18) [right of = 21] {$18$};
	\node[state, fill = green] (15) [right of = 14] {$15$};
	\node[state, accepting, fill = cyan, very thick] (25) [below of = 15] {$25$};

	\path[->]
	(1) edge node {a} (2)
	(2) edge node {b} (10)
	(1) edge node {b} (9)
	(9) edge node {a} (4)
	(4) edge node {b} (11)
	(11) edge node {a} (5)
	(5) edge node {d} (22)
	(1) edge node {c} (19)
	(19) edge node {a} (7)
	(7) edge node {a} (3)
	(3) edge node {b} (11)
	(19) edge node{b} (16)
	(16) edge node {b} (13)
	(13) edge node {b} (12)
	(12) edge node {a} (6)
	(6) edge node {d} (23)
	(19) edge node {d} (26)
	(26) edge node {a} (8)
	(8) edge node {c} (20)
	(20) edge [bend left] node {b} (17)
	(17) edge [bend left] node {c} (20)
	(17) edge node {b} (14)
	(14) edge node [label=above:d] {} (24)
	(26) edge [bend right] node [label=below:d] {} (27)
	(27) edge node {c} (21)
	(21) edge [bend left] node {b} (18)
	(18) edge [bend left] node {c} (21)
	(18) edge node {b} (15)
	(15) edge node {d} (25)
	;
	\end{tikzpicture}
	}
   \end{subfigure}
   \hfill \hfill
   \begin{subfigure}{0.15\textwidth}
       \centering
    \hspace{-20mm}   
    \scalebox{0.4}{
    \begin{tikzpicture}[shorten >=1pt,node distance=2cm,on grid,auto]
	\tikzstyle{every state}=[fill={rgb:black,1;white,10}]
	
	\node[rectangle, draw, orange] (1011) {$ (10, 11) $};
	\node[rectangle, draw, cyan] (23) [right of=1011] {$ (2, 3)$};
	\node[rectangle, draw, orange] (1314) [below of =1011] {$ (13, 14) $};
	\node[rectangle, draw, orange] (1617) [right of =1314] {$ (16, 17) $};
	\node[rectangle, draw] (2223) [below of = 1314] {$ (22, 23) $};
	\node[rectangle, draw, orange] (56) [right of =2223]{$ (5, 6) $};
	\node[rectangle, draw, cyan] (1112) [right of=56]{$ (11, 12) $};
	\node[rectangle, draw] (2425) [below of=2223]{$ (24, 25) $};
	\node[rectangle, draw] (1415) [right of=2425]{$ (14, 15) $};
	\node[rectangle, draw, cyan] (1718) [right of=1415]{$ (17, 18) $};
	\node[rectangle, draw, orange] (2021) [right of=1718]{$ (20, 21) $};
	\node[rectangle, draw] (34) [right of=23]{$ (3, 4) $};
	\node[rectangle, draw, orange] (45) [right of=34]{$ (4, 5) $};
	\node[rectangle, draw, orange] (67) [right of=45]{$ (6, 7) $};
	\node[rectangle, draw, orange] (1920) [right of=1617]{$ (19, 20) $};
	\node[rectangle, draw, orange] (78) [right of=1920]{$ (7, 8) $};
	\node[rectangle, draw, orange] (910) [right of=78]{$ (9, 10) $};
	\node[rectangle, draw, orange] (1213) [right of=1112]{$ (12, 13) $};
	\node[rectangle, draw, orange] (1516) [right of=1213]{$ (15, 16) $};
	\node[rectangle, draw] (2324) [right of=2021]{$ (23, 24) $};
	\node[rectangle, draw, orange] (2526) [below of=2425]{$ (25, 26) $};
	\node[rectangle, draw, orange] (2627) [right of=2526]{$ (26, 27) $};
	
	\path[->]
	(1011) edge [very thick, blue] node {} (23)
	(1314) edge [very thick, blue] node {} (1617)
	(1617) edge [very thick, blue] node {} (1920)
	(2223) edge [very thick, blue] node {} (56)
	(56) edge [very thick, blue] node {} (1112)
	(2425) edge [very thick, blue] node {} (1415)
	(1415) edge [very thick, blue] node {} (1718)
	(1718) edge [very thick, blue, bend left] node {} (2021)
	(2021) edge [very thick, blue, bend left] node {} (1718)
	;
	\end{tikzpicture}
	}
   \end{subfigure}
   \hfill
   \hspace{-23mm}
   \begin{subfigure}{0.38\textwidth}
   	\centering
	\scalebox{0.4}{
		\begin{tikzpicture}[shorten >=1pt,node distance=2cm,on grid,auto]
	\tikzstyle{every state}=[fill={rgb:black,1;white,10}]
	\node[state, initial, fill = purple] (1) {$1$};
	\node[state, fill= orange] (2) [above of=1] {$2$};
	\node[state, accepting, fill = green, very thick] (10) [right of = 2] {$10$};
	\node[state, fill = green] (9) [right of = 1] {$9$};
	\node[state, accepting, fill = green, very thick] (11) [right of = 4] {$11$};
	\node[state, fill = orange] (5) [right of = 11] {$5$};
	\node[state, fill = pink] (19) [below of = 1] {$19$};
	\node[state, fill = orange ] (7) [right of = 19] {$7$};
	\node[state, fill = orange] (34) [right of = 9] {$ \{3, 4 \} $};
	\node[state, fill = green] (16) [below of = 7] {$16$};
	\node[state, fill = green] (13) [right of = 7] {$13$};
	\node[state, accepting, fill = green, very thick] (12) [right of = 3] {$12$};
	\node[state, accepting, fill = orange, very thick] (6) [right of = 12] {$6$};
	\node[state, fill = orange] (8) [below of = 16] {$8$};
	\node[state, fill = cyan] (26) [left of = 8] {$26$};
	\node[state, fill = pink] (20) [below of = 13] {$20$};
	\node[state, fill = green] (17) [right of = 20] {$17$};
	\node[state, fill = green] (1415) [below of = 6] {$ \{14, 15 \} $};
	\node[state, fill = cyan] (27) [below of = 20] {$27$};
	\node[state, accepting, fill = pink, very thick] (21) [right of = 27] {$21$};
	\node[state, fill = green] (18) [right of = 21] {$18$};
	\node[state, accepting, fill = cyan, very thick] (22232425) [below of = 15] {$ \{ 22, 23, 24, 25 \} $};

	\path[->]
	(1) edge node {a} (2)
	(2) edge node {b} (10)
	(1) edge node {b} (9)
	(9) edge node {a} (34)
	(34) edge node {b} (11)
	(11) edge node {a} (5)
	(5) edge [bend left] node {d} (22232425)
	(1) edge node {c} (19)
	(19) edge node {a} (7)
	(7) edge node {a} (34)
	(34) edge node {b} (11)
	(19) edge node{b} (16)
	(16) edge node {b} (13)
	(13) edge node {b} (12)
	(12) edge node {a} (6)
	(6) edge [bend left] node {d} (22232425)
	(19) edge node {d} (26)
	(26) edge node {a} (8)
	(8) edge node {c} (20)
	(20) edge [bend left] node {b} (17)
	(17) edge [bend left] node {c} (20)
	(17) edge node {b} (1415)
	(1415) edge node  {d} (22232425)
	(26) edge [bend right] node [label=below:d] {} (27)
	(27) edge node {c} (21)
	(21) edge [bend left] node {b} (18)
	(18) edge [bend left] node {c} (21)
	(18) edge node {b} (1415)
	;
	\end{tikzpicture}
	}
   \end{subfigure}
   \caption{\footnotesize \emph{Left}: a sorted WDFA $ \mathcal{A}$ (node labels indicate the Wheeler order). States reached by the same label have the same color. \emph{Center}: the border graph $ \mathcal B(\mathcal A) $ built at Line 2 of Algorithm \ref{alg:minimize}. Nodes marked at Line \ref{algline:mark base} of the algorithm are orange, nodes marked at Line \ref{algline:mark reachable} are light blue. \emph{Right}: the minimum WDFA recognizing $ \mathcal{L(A)}$. Borders not marked (colored) in $ \mathcal B(\mathcal A) $ have been collapsed. }\label{fig:example}
\end{figure}
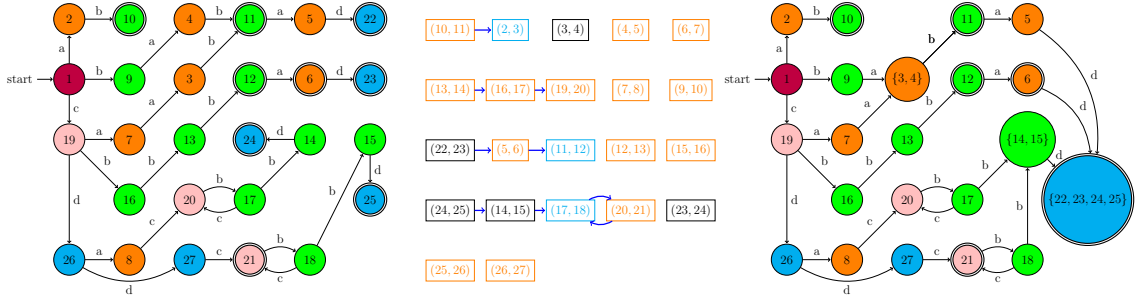

\Section{Experimental Results}

We implemented our algorithm and made the source available at the repository \url{github.com/nicolaprezza/dBg-min}.
Our tool takes as input a \texttt{fasta} or \texttt{fastq} dataset, builds the corresponding de Bruijn graph in BOSS format \cite{BOSS}, and runs our minimization algorithm on it. 
The tool also integrates an implementation of Alanko et al.'s strategy for pruning unnecessary dummy nodes \cite{alanko2021buffering} (disabled by default).
We tested our implementation on de Bruijn graphs of order $k=28$ built upon five real DNA datasets, see Table \ref{tab:datasets}: three collections of genomes downloaded from the Pizza\&Chili corpus (\url{pizzachili.dcc.uchile.cl/repcorpus.html} ---  \emph{Saccharomyces cerevisiae}, \emph{Haemophilus influenzae}, and \emph{Saccharomyces paradoxus}), and short reads sequenced from \emph{Escherichia coli} (\url{www.ebi.ac.uk/ena/browser/view/ERR022075}) and \emph{Human} (\url{www.ebi.ac.uk/ena/browser/view/SRX7829390}) genomes. 
All experiments were run on a workstation with 128 GiB of RAM and a Intel(R) Xeon(R) W-2245 CPU, running Ubuntu Linux and using one thread.

The three \texttt{fasta} datasets contain just one sequence and thus the BOSS representation \cite{BOSS} of their de Bruijn graph has at most $k$ dummy nodes. As a result, the pruning algorithm of Alanko et al. \cite{alanko2021buffering} does not compress these graphs. The remaining two \texttt{fastq} datasets, on the other hand, generate de Bruijn graphs that, as we show in Table \ref{tab:results}, are already compressible with Alanko et al.'s algorithm. Table \ref{tab:results} shows the results we obtained by running our WDFA minimization algorithm on the de Bruijn graphs built from the original datasets, and --- just on the \texttt{fastq} datasets --- on the de Bruijn graphs pruned with Alanko et al.'s algorithm.

\begin{table}[h!]
\small
    \centering
    \begin{tabular}{|c|c|c|}
    \hline
    \texttt{dataset} & \texttt{seqs} & \texttt{bases} \\\hline
    cere.fasta & 1 & 461,286,644  \\
    influenza.fasta & 1 & 154,804,605 \\
    para.fasta & 1 & 429,265,758 \\
    ecoli.fastq & 45,440,200 & 4,589,460,200 \\
    human.fastq & 63,917,134 & 6,455,630,534  \\
    \hline
    \end{tabular}
    \caption{\footnotesize Datasets on which we built the de Bruijn graphs used as inputs for our algorithm. The columns show, respectively: name of the dataset, number of input sequences, and number of DNA bases (i.e. characters on the alphabet $\{A,C,G,T\}$).}
    \label{tab:datasets}
\end{table}

\begin{table}[h!]
\small
    \centering
    \begin{tabular}{|c|c|c|c|c|c|}
    \hline
    \texttt{dataset}  & \texttt{in ($\times 10^6$)} & \texttt{out ($\times 10^6$)} & reduction & \texttt{time (s)} & nodes/s ($\times 10^6$) \\\hline
    cere.fasta & 19.004 & 15.756 & 17.1\% & 17 & 1.118 \\
    influenza.fasta & 6.469 & 4.792 & 25.9\% & 5 & 1.294 \\
    para.fasta & 28.178 & 22.556 & 19.9\% & 26 & 1.084 \\
    ecoli.fastq & 449.92 & 220.47 & 51\% & 398 & 1.130\\
    human.fastq & 650.51 & 438.68 & 32.6\% & 600 & 1.084 \\
    \hline\hline
    ecoli\_pruned & 317.173 & 201.940 & 36\% & 291 & 1.089 \\
    human\_pruned & 449.991 & 387.627 & 13.8\% & 431 & 1.044 \\\hline
    \end{tabular}
    \caption{\footnotesize Performance of our minimization algorithm on the de Bruijn graphs built over the datasets of Table \ref{tab:datasets} (first four rows) and on the pruned de Bruijn graphs \cite{alanko2021buffering} of the two fastq datasets (last two rows). The order of all de Bruijn graphs is $k=28$. The columns show, respectively: name of the dataset, number of nodes in the input de Bruijn graph, number of nodes in the output (minimized) de Bruijn graph, percentage of nodes removed by the minimization algorithm, running time (construction of the de Bruijn graph in BOSS format is not counted towards the running time), and number of processed nodes per second.}
    \label{tab:results}
\end{table}

\paragraph{Discussion} Our implementation proved to be extremely fast, processing over one million nodes per second. The results show that WDFA minimization is indeed a relevant compression strategy for de Bruijn graphs: our algorithm reduced the number of nodes of the original de Bruijn graphs from 17.1\% to 51\%, and from 13.8\% to 36\% when the graph was previously pre-processed with the pruning algorithm of Alanko et al. \cite{alanko2021buffering}. Interestingly, these results indicate that, while the two algorithms target mostly different sources of redundancy, minimization targets some dummy nodes as well. The combination of the two algorithms reduced the number of nodes by 55.1\% on the \texttt{ecoli.fastq} dataset, and by 40.4\% on the \texttt{human.fastq} dataset. We leave it as an interesting open problem to show whether or not the combination of the two algorithms is optimal, in the sense that it generates the smallest WDFA containing all (and only the) labeled paths of the input de Bruijn graph. 

\bibliographystyle{IEEEbib}
\bibliography{main}
\end{document}